\documentclass[11pt]{article}
\usepackage{amsfonts,amsmath}
\usepackage{cite}

\usepackage{amssymb,amsmath,amsfonts,amsthm,enumerate}

\usepackage{latexsym}

\usepackage{color}
\usepackage{graphics}
\usepackage{graphicx}
\usepackage{indentfirst}

\usepackage{algorithm}
\usepackage{algorithmic}
\usepackage{epsfig}
\usepackage{subfigure}
\usepackage{enumerate}
\usepackage{amsmath}

\newtheorem{theorem}{Theorem}[section]
\newtheorem{lemma}{Lemma}[section]
\newtheorem{claim}{Claim}[section]
\newtheorem{cor}{Corollary}[theorem]

\newcounter{rem}
\newcommand{\ignore}[1]{ }

\def\mPr{{\bf P}}

\begin{document}

\title{Approximating Quadratic 0-1 Programming via SOCP}

\author{Sanjiv Kapoor\footnote{Department of Computer Science, Illinois Institute of Technology, Chicago, IL 60616. E-mail: {\tt kapoor@iit.edu}} and Hemanshu Kaul\footnote{Department of Applied Mathematics, Illinois Institute of Technology, Chicago, IL 60616. E-mail: {\tt kaul@math.iit.edu}}
}


\maketitle

\begin{abstract}
We consider the problem of approximating
Quadratic O-1 Integer Programs with bounded number
of constraints and  non-negative
constraint matrix entries, which we term as PIQP.

We describe and analyze a randomized algorithm based on a
program with hyperbolic constraints (a Second-Order Cone Programming -SOCP- formulation)
that achieves an approximation ratio of
$O(a_{max} \frac{n}{\beta(n)})$, where $a_{max}$ is the maximum
size of an entry in the constraint matrix and  $\beta(n)  \leq \min_i{W_i} $, where
$W_i$ are the constant terms that define the constraint inequalities.
We note that by appropriately choosing $\beta(n)$
the  randomized algorithm, when combined with other algorithms that achieve good
approximations for smaller values of $ W_i$, allows better algorithms for  the complete
range of $W_i$.
This, together  with a greedy algorithm, provides a
$O^*(a_{max} n^{1/2} )$ factor approximation,
where $O^*$ hides logarithmic terms.
Our solution is achieved by a randomization of
the optimal solution to the relaxed version
of the hyperbolic program. We show that this solution provides the approximation
bounds using concentration bounds  provided by Chernoff-Hoeffding and Kim-Vu.

\end{abstract}

\section{Introduction}

In this paper, we study optimizing 0-1 integral quadratic programs.
We consider the class of
quadratic integer programs where $X \in \{ 0,1 \}, X= (x_1 , \ldots x_n)$:
\begin{eqnarray*}
\max  \  X^T B X &+& c^TX\\
{\rm subject \ to}\\
  a_i^TX &\leq& W_i, \ \ i=1, \ldots p \\
 X_i &\in& \{ 0, 1 \}
\end{eqnarray*}
for bounded number of constraints $p$. We assume that
the quadratic objective  function is defined by a symmetric matrix $B$ with
non-negative coefficients and  the linear term is defined by a positive vector
$c$. Furthermore, $a_{ij} \geq 0 , i=1 \ldots p, j= 1 \ldots n$.
We term the above as a {\em Positive 0-1 Quadratic Program} (PIQP).
The problem is NP-hard and our interest is in designing approximation schemes
for this problem.

While Quadratic Integer programs have been considered in the operations  research community \cite{KimKojima2001},
approximation  algorithms for the general problem have not been actively considered.
A recent result for positive semi-definite matrices by Guruswami and Sinop\cite{GS2011},
provides an approximation scheme for quadratic integer  programs with PSD objective
for graph problems
with a bound of $\frac{(1+\epsilon)}{\min ( 1, \lambda_r)}$ and a running time of
$n^{O(r^*/\epsilon^2)}$, where $r^*$ is the number of eigenvalues of the graph
Laplacian  smaller than $1-\epsilon$.  Unfortunately we do not see any way to bound $r^*$ in general.

An important subclass of the PIQP is the generalized knapsack problem.
The  motivation is a resource allocation problem, where the items are proposed projects,
the weight of each item is the cost of implementation, and
the benefit measures the gains due to implementation of that project.
Typically, this budgeting problem is implemented either as a 0-1 Knapsack problem which assumes that the projects are pairwise independent so their benefits are additive, or all possible combinations of projects are considered which increases the size of the problem exponentially. However, in most applications the benefit of two projects can be less than or greater than the sum of the individual projects due to underlying dependencies among the projects. The current authors with collaborators have explored these ideas of interdependency between projects in the context of transportation resource allocation \cite{LiKK11}. This generalized knapsack model was applied in a computational study using real-life data on travel demand, roadway network designs, and traffic operations, as well as six major projects proposed for possible investments to the Chicago downtown area. The computational study revealed that the overall benefits in terms of travel-time savings after considering project network-wide impacts and their interdependency relationships are much lower (by 38-64 percent) as compared with those established without considering project interdependency relationships, indicating an inflated benefits when interdependencies are not considered.

To model this situation, we define a problem that will
be critical to developing our approximation algorithm.
The {\em Graph Knapsack Problem}, $GKP(G,b,w,W)$,
where $G = (V, E)$ is an undirected graph
with $n$ vertices,  $w: V \rightarrow \mathbb{Z}^+$ is a weight function,
$b: E \cup V \rightarrow \mathbb{Z}$ is a benefit function on vertices and edges, and  $W$ is a weight bound.
The vertices correspond to the items in the Knapsack problem. The benefit of a subgraph $H=(V_H, E_H)$ is $b(H)= \sum_{v \in V_H} b(v) + \sum_{e \in E_H} b_{e}$ while its weight is $ w(H)= \sum_{v \in V_H} w(v)$. Note that the benefits can be negative;
negative weight edges model the case where two projects' benefits are literally less than the sum of their parts.
Given a  subset of vertices $S$, we consider the subgraph induced by $S$, termed $G[S]$.
The graph knapsack problem requires the determination of a subset of vertices $S \subseteq V$ that maximizes the benefit of the induced subgraph, $b(G[S])$ with the restriction that its weight $w(G[S])$ is less than $W$. Note that this reduces to the classical  knapsack problem when there are no edges in the graph $G$.
Clearly GKP is  $NP$-Hard.

From a graph theoretic point of view, GKP is related to the {\em maximum clique problem}.
We can reduce the clique problem to the graph-knapsack problem.
Given a graph $G$, suppose we wish to determine if $G$ contains a
clique of size $t$. We define an instance of GKP on $G$
with $W=t$, $w_i =1$, $b_i = 0$, $b_e=1$ for $e\in E(G)$.
Graph $G$ has a $K_t$ iff  GKP has benefit at least ${t \choose 2}$.

We may note that, unless $P=NP$, achieving an approximation ratio
better than $n^{1-\epsilon}$ is
impossible for the clique problem \cite{Hastad99,Zuckerman05}. However, it is not known whether there exist approximation-preserving transformations between GKP and Max-Clique.

In fact, GKP also generalizes {\em the Dense $k$-Subgraph problem} ($k$-DSP problem)
(see \cite{FKP2001,SW1998}), which requires finding an $k$-vertex induced subgraph of an edge-weighted graph with maximum density.
This corresponds to GKP with edges of benefit 1 while vertices have zero benefit,
and the weight of each vertex is 1 with $W=k$. The dense $k$-subgraph problem is well-researched problem, with the best approximation provided
by an  $O(n^{1/4})$ approximation factor algorithm in \cite{BCCFV2010}. Algorithms that use semi-definite programming have been proposed by Goemans (as mentioned in \cite{FKP2001}) and in \cite{S99findingdense}. These algorithms promise an approximation ratio of $O(n/k)$. Note that a PTAS for the dense $k$-subgraph has been ruled out in \cite{Khot} under a certain complexity assumption.

\ignore{
The dense $k$-subgraph problem is well-researched problem, with the best approximation provided
by an  $O(n^{1/4})$ approximation factor algorithm in \cite{BCCFV2010}.
Prior work includes an approximation factor of $n^{\delta}, \delta < 1/3$ that
has been achieved in \cite{FKP2001}.
Allowing general non-negative benefits  on edges has
also been considered, adding a multipicative factor of $\log n$ to the approximation.
Note that a PTAS for the dense $k$-subgraph has been ruled out in \cite{Khot}
under a certain complexity assumption.
Algorithms that use semi-definite programming have been proposed by Goemans \cite{FKP2001}
and in \cite{S99findingdense}. These algorithms find approximation that are a constant factor
of the optimal for large value of $k$. In fact,  they promise an approximation ratio
of $O(n/k)$. A  PTAS can be obtained when $k= \Omega(n)$ and the number of
edges is $\Omega(n^2)$ \cite{AKK95}.
Other dense subgraph problems have also been considered in \cite{KS2009,Charikar00}.\\
}

The problem can be  formulated as a $0$-$1$ Quadratic Program:
$$\begin{array}{rlll}
&maximize & \sum_i b(v_i)x_i\ \ + \sum\limits_{i,j:v_iv_j \in E(G)} b(v_iv_j)x_ix_j\\
&\text{such that}&
\sum_i w(v_i)x_i \leq W&\\
&&x_i\in\{0,1\}&\\
\end{array}$$
Replacing the term $x_ix_j$ by an integer variable $x_{ij} \in \{0, 1 \}$ and adding the constraints $x_{ij} \leq \frac{x_i + x_j}{2}$ and $x_{ij} \geq \frac{x_i + x_j-1}{2}$ also gives an integer linear program (ILP) for the problem. Without loss of generality, we can assume that $G=K_n$, since we can give non-edges benefit zero.
This quadratic program is equivalent to the well-studied Quadratic Knapsack Problem (QKP) \cite{GHS1980} (see~\cite{P2007} for a survey):
$$\begin{array}{rlll}
&maximize & \sum\limits_{i=1}^{n}\sum\limits_{j =1}^{n} b_{ij}x_ix_j\\
&\text{such that}&
\sum\limits_{i=1}^n w_ix_i \leq W&\\
&&x_i\in\{0,1\}&\\
\end{array}$$
(To account for $b(v_i)$, add a variable $x_{i'}$ with $b_{ii'}=b(v_i)$ and $w_{i'}=0$.)  However, the benefits in QKP are typically assumed to be non-negative~\cite{P2007}.

The QKP is an important problem which has been mostly studied from the LP-based exact algorithms point of view~\cite{P2007,GHS1980}. Rader and Woeginger \cite{RW2002} developed a FPTAS for the case when all benefits are non-negative and the underlying graph is series parallel. They also show that when QKP has both negative and non-negative benefits, it cannot have a constant factor approximation unless $P=NP$.


The idea of using discrete structures like graphs, digraphs, or posets to generalize the classical knapsack problem by modeling some sort of dependency among the items is not a new one. However all such generalizations of the Knapsack problem (described below)
restrict the choice of subset of items that can be picked.

\ignore{
While our model does not restrict the choices directly, instead it modifies
the benefit function so that the benefit on the edge between a pair of items
could act as a penalty (if its negative) or an inducement (if its positive)
towards the choice of those two items.
Appropriate choice of benefits on the edges can mimic some variants of the Knapsack Problem.
}

The {\em Knapsack Problem with Conflict Graph} (KPCG) is a knapsack problem where each edge in the underlying conflict graph on the items introduces the constraint that at most one of those two items can be chosen. This can be modeled as the Graphical Knapsack problem by putting large negative benefit on the edges of the conflict graph and using that as the underlying graph for GKP.
The KPCG was first studied by Yamada et al. \cite{YKW2002} who analyzed some greedy algorithms for KPCG.
Pferschy and Schauer~\cite{PS2009} give exact dynamic programs which can be implemented as FPTAS when the underlying conflict graph is a chordal graph or has bounded tree width.


Borradaile, Heeringa, and Wilfong~\cite{BHW2009} study two versions of
the {\em Constrained Knapsack Problem} (CKP) in which dependencies between items are given by a graph. In the first version, an item can be selected only if at least one of its neighbors is also selected. In the second version, an item can be selected only when all its neighbors are also selected.
They give upper and lower bounds on the approximation ratios for both these problems on undirected and directed graphs.

Similarly, the {\em Precedence-Constrained Knapsack Problem} (PCKP) introduces a directed acyclic graph (or, equivalently a poset) on the items so that a item can be chosen if and only if all of its predecessors have already been chosen (selected items should form an ideal in the poset).
See Kolliopoulos and Steiner~\cite{KS2007}, Samphaiboon and Yamada~\cite{SY2000}, Kellerer, et.~al~\cite{KPP2004} for further discussion.

In the {\em Subset-Union Knapsack Problem} (SUKP), each item (with benefit) is a subset of a ground set of elements (each with a weight). SUKP asks to find a set of elements with maximum benefit that fits in the knapsack. This corresponds to HKP with hyperedges corresponding to a item in SUKP and vertices corresponding to the elements of SUKP. Goldschmidt, et. al~\cite{GNY1994} show that SUKP is NP-hard and give an exact algorithm that runs in exponential time.

\subsection{Our Results and an Outline}

In this paper we describe a polynomial-time randomized algorithm that approximates
PIQP to within a factor $O(a_{max} n^{1/2} \log^{2+\gamma} n)$  w.h.p, when the benefit
function $b$ is non-negative, $\gamma$ is a very small constant and $a_{max}$ is the
maximum co-efficient in the constraint matrix.


The approximation to the PIQP problem results from a combination of two algorithms, one which works well for small values of $W=\max_i \{ W_i \}$ with an $O(\min(n,W)/t)$-approximation bound (for any fixed $t>0$), and the other a randomized algorithm which gives a $O(a_{max} \frac{n}{\beta(n)} \log^{2 + \gamma}n)$ approximation bound where $a_{max}$ is the maximum size of an entry in the constraint matrix and  $\beta(n)  \leq \min_i{W_i} $, where
$W_i$ are the constant terms that define the constraint inequalities and $\gamma$ is a small constant.

\ignore{
We apply the above algorithm  to
give an approximation of $O^*(n^{1/(1+t)} )$ for the $k$-DSP problem for sparse graphs
having  maximum degree at most $n^{1/t}$, for any integer $t >1$.
The $O^*$ notation is used in the standard way to hide logarithmic terms.
}

Our primary aim in this paper is to illustrate the use of the following main tools as compared to the typical application of Semi-definite Programs and Chernoff bounds in problems of this nature:
\begin{enumerate}

\item A second-order cone program  (a convex program with hyperbolic constraints).
While a similar result can be obtained via general semi-definite programming,
SOCP has more efficient solutions (see the discussion in ~\cite{lbvl98}, where it is said that ``solving SOCPs via SDP is not a good idea, interior-point  methods that solve the SOCP directly have a much better worst-case complexity than a SDP method'').

\item A combination of the Kim-Vu bounds on concentration of polynomial random variables, and Chernoff-Hoeffding concentration bounds.

\end{enumerate}

Section~2 presents a simplified version of PIQP, and some preliminaries for the analysis of the algorithm.
Section~3 presents an analysis of a natural greedy method that gives an $O(\min(n,W)/t)$-approximation algorithm (for any fixed $t>0$).
In Section~4, we give the complete description and analysis of our main algorithm with $O(a_{max} \frac{n}{\beta(n)} \log^{2 + \gamma}n)$ approximation. This bound can be thought as a generalization of the $O(n/k)$-approximation algorithm for the $k$-DSP (~\cite{FKP2001}, ~\cite{S99findingdense}), although our methodology is different.
In Section~5 we combine the algorithms from sections~3 and 4 to provide the general approximation algorithm.

The appendix contains an algorithm for the  solving the generalized knapsack problem in the case of bounded treewidth. Since series-parallel graphs have tree-width at most~2, our algorithm generalizes the one from Rader and Woeginger~\cite{RW2002}.

\section{Preliminaries}

\subsection{Preliminaries-A}
We study the following class of quadratic integer programs where $X \in \mathbb{Z}^n, X= (x_1 , \ldots x_n)$:
\begin{eqnarray*}
\max  \  X^T B X &+& c^TX\\
{\rm subject \ to}\\
  a_i^TX &\leq& W_i, \ \ i=1, \ldots p \\
 x_i &\in& \{ 0, 1 \}
\end{eqnarray*}
for bounded  number of constraints $p$ and
positive integer vectors $a_i$ and $c$ and the
matrix $B \in \mathbb{Z^+}^{n \times n}$.
We term the above as a {\em Positive Integer Quadratic Program} (PIQP).
W.l.o.g we will assume that $B$ is a symmetric matrix.

The objective function
has both quadratic and linear terms.
We can interpret optimizing the  linear term
as the classical 0-1 Knapsack problem with multiple constraints for which there is
a PTAS \cite{books/knapsackProblems}.
Separating the linear and quadratic terms  will introduce a constant factor in the
approximation bound. Hence
we will restrict our attention to the quadratic objective
function $\sum\limits_{i \neq j} b_{ij}x_ix_j$.


For simplicity we will
transform the above inequalities by appropriate scaling
such that  $ \hat{a}_i^TX \leq W, \ \ i=1, \ldots p$,
where $W = \max_i \{W_i \}$ and  $\hat{a}_{ij} = \lceil a_{ij}W/W_i \rceil$.
Note that $\hat{X}_i \cdot \frac{a_{ij}W/W_i}{\lceil a_{ij}W/W_i \rceil}$ is
a solution to the original problem iff $\hat{X}$ is a solution to the scaled problem.
Let $\hat{x}^*_{ij}$ and $x^*_{ij}$
be the optimum solution to the scaled and original problems, respectively.
Since $\frac{\lceil a_{ij}W/W_i \rceil}{a_{ij}W/W_i} \leq 2$,
we get that $x^*_{ij}/2$ is a solution to the scaled problem.
We thus obtain that  the optimum solution to the scaled problem
is within a factor of $4$ of the optimum solution to the
original problem.

Furthermore,
we can assume that $a_{ij} + a_{ik} \leq W$ for every pair of indices $(j,k)$. 
Otherwise, the benefit $b_{ij}$ can be set to $0$ as it is not possible to set 
both variables $x_i$ and $x_j$ to $1$. 
Thus we can assume that $x_k$ and $ x_l$ are such 
that $(x_k,x_l) = \arg \max_{(x_i,x_j)} b_{ij}$ can be set to $1$ . 

The transformed problem, which we term PIQP also, is: \\


\begin{eqnarray*}
\max  \  \sum\limits_{i,j} b_{ij}x_ix_j\\
{\rm subject \ to}\\
  a_i^TX &\leq&  W e, \ \ i=1, \ldots p \\
 x_i &\in& \{ 0, 1 \}
\end{eqnarray*}
where $e$ is the vector $(1, 1 \ldots ,1)$. We also denote
the matrix $[a_{ij}]$ by $A$, with the $i$th  row of the matrix $A$
corresponding to the constraint specified by the vector $a_i$.

We next define two restrictions of the PIQP  problem.
Let $X_A = \{ x_j | \exists i , \  a_{ij} > W/2 \}$
\begin{itemize}

\item{PIQPS Problem:}
This problem is a version of the PIQP under the assumption
that $a_{ij} \leq W/2$. This gives
us a simplified version of the PIQP problem by assigning all variables
in $X_A$ to be of value $0$.

\item{PIQPR Problem:}
This problem is obtained by  considering PIQP
with the additional requirement that for every variable $x_j$ there
exists a co-efficient in the constraint matrix, $a_{ij} > W/2$.
This gives
us a version of the problem where only one variable in each constraint
can be set to be $1$. In this problem all variables in $X-X_A$ are set  to be zero.

\end{itemize}

We justify this choice of  additional restrictions below:
\begin{claim}
\label{PIQPSRclaim}
An $\alpha$-approximation to $PIQPR$ and $PIQPS$ problems gives an
$\alpha$-approximation to the PIQP problem.
\end{claim}
\begin{proof}
Let $X_A = \{ x_j | \exists i , \  a_{ij} > W/2 \}$
The optimal solution to PIQP, $OPT$ can be partitioned into two subset of variables,
$OPT= S1 \cup S2, S1 \subseteq X_A, S2 \subseteq X- X_A$ (where by overloading notations,
$X$ is used to denote the set of all the variables),
that are assigned the value $1$.
The benefit obtained
by $S1$ is bounded by $OPT(PIQPR)$ and the benefit obtained by $S2$ is bounded
above by $OPT(PIQPS)$. The benefit provided by one of $S1$ and $S2$  alone
is at least half of the
optimum solution.  Thus  $\max \{ OPT(PIQPS) , OPT(PIQPR) \} \geq OPT/2$
Thus an $\alpha$-approximation to $PIQPS$ and $PIQPR$ gives an $\alpha$-approximation
to $PIQP$.
\end{proof}

Furthermore, we have the following property of any maximal solution to PIQPS, i.e.
a solution where no additional variable can be set to value $1$ without violating
the constraints:
\begin{lemma}
\label{everyW/2}
Every maximal solution to PIQPS  is either the trivial solution where
$x_j =1 , \forall j$ or the solution uses at least $W/2$ of the budget in
at least one of the constraints.
\end{lemma}
\begin{proof}
We know  by assumption that all coefficients $a_{ij}$ are less than $W/2$.
\ignore{
Suppose not, i.e. $a_{lj} > W/2$ for some $l,j$.
Let $B(A) = \{ x_j | \exists i , \  a_{ij} > W/2 \}$.
Let ${\cal S}(B)$ be the set of solutions generated by picking some subset
of variables in $B(A)$ to be $1$ and all other variables zero, i.e.
\[{\cal S}(B) = \{ S(B)=\{x_{j1}, x_{j2} \ldots x_{jl} \} | S(B) \subseteq B(A) ,
AX \leq W , \ \math{where}  \  x_i = 1 \math{iff} x_i \in S
\}
\]
Note that  at most one variable per inequality $\sum_{j} a_{ij}x_j \leq W$
can be set to $1$ in each of the solutions corresponding to $S(B)$.
}

Now suppose there exists a solution, $S$, such
that $\sum_{j} a_{ij}x_j < W/2,  \ \forall i$.
Let $S = \{ x_j | x_j =1 \}$.
Then  assuming that all variables are not already assigned the value $1$
we can assign another variable $x_k \notin S$
to be $1$ without violating any of the inequalities. This proves the lemma.
\end{proof}

\subsection{Preliminaries-B}

Since we will be dealing with non-linear random-variables, an important tool
for us is the polynomial concentration bound of
Kim-Vu \cite{kimvu:concentration,vu:con2} that applies to any multi-linear random
variables defined on a underlying hypergraph with edges of bounded size. For
our application, the underlying hypergraph will be 2-uniform, that is, it
will be a simple graph with no loops.

Let $Y = \sum\limits_{e\in E(G)} b_{e} \prod\limits_{v\in e} y_v$ where $G$ is
an underlying graph, $b$ is non-negative weight function on the edges,
and $y_v$ are independent $0-1$ random variables.

Define $\varepsilon_0 = {\bf E}[b(Y)]$, $\varepsilon_1 = \max_v (\sum_{u\in N_G(v)} b_{uv} {\bf P}[y_u =1])$, and $\varepsilon_2 = \max_{uv\in E(G)} b_{uv}$.

\begin{theorem}{[J-H. Kim, Van Vu, 2001]}\\
${\bf P}[{\bf E}[Y] - Y > t^2] < 2e^2 \;\; exp\left(-t/32(2\varepsilon\varepsilon')^{1/4} + \log n\right)\;\;$ where $\varepsilon= \max \{\varepsilon_0, \varepsilon_1, \varepsilon_2 \}$, and $\varepsilon'= \max \{\varepsilon_1, \varepsilon_2 \}$.
\end{theorem}

We will also use the following standard bounds of Chernoff and Hoeffding~\cite{C52,H63}:

\begin{theorem}{[Chernoff-Hoeffding Bound-1]}
Let $X_1 , X_2 , \ldots X_n$ be independent $0$-$1$  random variables.
Denote $X= \sum_{i=1}^n a_i X_i$, where $a_i \in [0,1]$, and let $E(X)=\mu$. Then
${\bf P}[X \geq (1+\delta)\mu] \leq e^{-(\mu \delta^2)/3} $ when $0 < \delta \leq 1$.



\end{theorem}
and
\begin{theorem}{[Chernoff-Hoeffding Bound-2]}
Let $X_1 , X_2 , \ldots X_n$ be independent $0$-$1$  random variables.
Denote $X= \sum_{i=1}^n  X_i$,
and let $E(X)=\mu$. Then
${\bf P}[X- \mu \geq t) \leq e^{-2nt^2} $.
\end{theorem}

For our purpose, as we will be dealing with weighted sums, the following form of Chernoff-Hoeffding bound is appropriate.

\begin{cor}
Let $X_1 , X_2 , \ldots X_n$ be independent $0$-$1$  random variables.
Denote $X= \sum_{i=1}^n a_i X_i$, where $a_i \geq 0$, and let $E(X)=\mu$. Then, for $a_{max} = \max \{a_i\}$,
$ {\bf P}[X \geq (1+\delta)\mu] \leq e^{-(\mu \delta^2)/(3 a_{max})} $ when $0 < \delta \leq 1$.


\end{cor}

\section{A Greedy Approach}

\def\cO{{\cal O}}

\ignore{
Given an instance of a convex quadratic program of the following form:
\begin{eqnarray*}
\max  \  \sum\limits_{i,j:x_ix_j \in E(G)} b(x_i,x_j)x_ix_j\\
{\rm subject \ to}\\
  a_i^TX &\leq& W_i, \ \ i=1, \ldots p \\
 x_i &\in& \{ 0, 1 \}
\end{eqnarray*}
for bounded number of constraints $p$.

For simplicity, we write the objective function as $\sum_i b(v_i)x_i\ \ + \sum\limits_{i,j:v_iv_j \in E(G)} b(v_iv_j)x_ix_j$ and the linear constraints as $\sum_{i,j}a_{ij}x_j \leq W, \; i=1, \ldots p $. Note that we have changed the the budget bounds on the right to a uniform $W$ by straightforward scaling. Observe that $p=1$ corresponds to QKP. The support of function $b$ can thought of as the underlying graph $G$.

As discussed for QKP, we will restrict our attention to the quadratic portion of the objective function. The problem of optimizing the problem with b(e) = 0 for all edges is simply the classical 0-1 knapsack problem with $p$ budget constraints which can be solved using a FPTAS for bounded $p$. This changes the overall approximation factor by at most a constant factor. So we will assume the objective function is $\sum\limits_{i,j:v_iv_j \in E(G)} b(v_iv_j)x_ix_j$.
}

%
%

In this section we consider the two restricted problems $PIQPS$ and $PIQPR$, 
as defined in Section~2. 
\begin{eqnarray*}
\max  \  \sum\limits_{i,j} b_{ij}x_ix_j\\
{\rm subject \ to}\\
  a_i^TX &\leq&  W e, \ \ i=1, \ldots p \\
 X_i &\in& \{ 0, 1 \}
\end{eqnarray*}
where $e$ is the vector $(1, 1 \ldots ,1)$.
In PIQPS, $a_{ij} \leq W/2$ $\forall i,j$.
And in PIQPR $\forall i, \  \exists j \ \mbox{s.t.} a_{ij} > W/2$.  
In the PIQPS problem we will assume that the trivial solution where  every variable
is set to $1$ is not  possible else we can trivially return the optimal solution.

We will interpret this quadratic optimization program with zero diagonal entries as a
directed graph optimization problem, $GKPM$ (Generalized Knapsack with Multiple Constraints), 
with each variable $x_i$ corresponding to a vertex $v_i$ and $b_{ij}$ 
corresponding to the benefit on the edge $v_iv_j$.
The support of function $b$ can be thought as the edges in the
underlying graph $G$. The coefficient $a_{ij}$ is the weight associated with $v_j$ in the $i$th constraint.

We define a combined weight for each vertex $i$ as $w_i = \sum_j a_{ji}$, the sum of all 
the weights in column $i$ (associated with the variable $x_i$) in all the budget constraints. 
Let $W$ be as specified in the problems PIQPS and PIQPR. 
Note that this gives us a multi-constraint version of an instance of $GKP(G,b,w,W)$ as defined in the introduction, which we term $GKPM$.

Define $w(T) = \sum_{v_i \in T} w_i$ for every subset of vertices $T$ for the rest of this discussion.

Now, using this definition of weight, the greedy algorithm can be defined naturally as:
Fix an integer $t \geq 2$.
\begin{enumerate}
\item Initialize $S = \emptyset$.
\item Pick a subset $T$ of $V(G)-S$ of cardinality at most $t$ such that its benefit
(the sum of the benefits of the  edges incident to vertices in $T$ in $S \cup T$)
to weight (sum of the combined weights, $\sum_{i \in T} w_i$) ratio is highest.
\item Update $S = S \cup T$ if weight of $S \cup T$ satisfies each of the $p$ budget constraints, and then go to step 2. Otherwise pick whichever of $S$ or $T$ has larger benefit as the final solution as long as the set  satisfies all the budget constraints.
\end{enumerate}

Note that any final solution given by this algorithm will satisfy all the constraints.
This follows since for any choice of $T$ it is easy to check the feasibility of $T$ and
$S \cup T$.

The following is true from Lemma~\ref{everyW/2} and by the  definition of PIQPR.
\begin{claim} Every solution generated by the greedy algorithm  for
PIQPR and PIQPS uses at least $W/2$ of the budget in at least one the constraints.
\end{claim}

\ignore{
\begin{proof}
We will assume w.l.o.g that all coefficients $a_{ij}$ are less than $W/2$.
Suppose not, i.e. $a_{lj} > W/2$ for some $l,j$.
Let $B(A) = \{ x_j | \exists i , \  a_{ij} > W/2 \}$.
We can solve a reduced problem by picking a subset of variables
$S_B= \{ x_j | x_j \in B(A) \}$  which are feasible,
i.e. $ \sum_{j} a_{ij} \leq W , \ \forall i$.
Note that this implies that at most one variable per inequality $\sum_{j} a_{ij}x_j \leq W$
can be set to $1$.
The number of such sets is bounded by $C(n,p)= O(n^p)$. Assigning value $1$ to
variables in feasible sets $S_B$ and $0$ to other variables in $B(A)$ results in
a problem where all coefficients are now $ \leq W/2$. There are $O(n^p)$ such problems
and the best of these solutions gives an optimal solution.

For problems with $a_{ij} \leq W/2, \forall i,j$, suppose there exists a solution such
that $\sum_{j} a_{ij}x_j \leq W/2,  \ \forall i$.
Then  assuming, naturally, that all variables are not already assigned,
we can assign another variable
to be $1$ without violating any of the inequalities. This proves the Claim.
\end{proof}
}

For an instance $I$ of $GKPM$, let $A(I)$ be an approximate solution  generated by the greedy method
and let $C(I) = A(I) \cap \cO(I)$, where $\cO(I)$ is an optimal solution.
W.L.O.G. we assume that $G$ is a complete graph by defining $b(e)=0 , \forall e \notin E(G)$.

Note that in general $A(I)$ and $\cO(I)$ will overlap and this intersection between the solutions makes it difficult to analyze the relative worth of the two solutions. To overcome this difficulty we use $A(I)$ and $\cO(I)$ to define a new instance of $GKPM$ which has disjoint greedy and optimal solutions with its greedy solution same as $A(I)$. This makes it easier to analyze and indirectly bound the benefits of the original  $A(I)$ and $\cO(I)$.


We define a transformation to an instance $I' = (G',b',w',W)$ as follows:
Define a  new set of vertices $C'= \{ v' | v \in C(I)\}$, i.e., the vertex set $C'$ is a 
duplicate of $C(I)$. 
In instance $I'$, we let $A=A(I)$ and $O'= (O(I) \setminus C(I)) \cup C'$
Note that $A(I) \cup \cO' \setminus C' = A(I) \cup \cO(I)$.
\begin{enumerate}
\item
$V(G') = V(G) \cup C'$

\item The benefits for the vertices are all zero (as illustrated above).

\ignore{
assigned as:
\[
b'(v) = \left\{
\begin{array}{l l}
b(v) & \mbox{ if $v \in A(I) \cup \cO'\setminus C'$} \\
0 & \mbox{otherwise} \\
\end{array}
\right.
\]
}

\item
The benefits for the edges are
\[
b'(e) = \left\{
\begin{array}{l l}
b(e) & \mbox{ if $e= uv \text{ s.t.} \ u,v \in A(I)$} \\
b(e) & \mbox{if $e=uv \text{ s.t.}  \ u \in \cO' \text{ and } v \in \cO'\setminus C', \text{ or }  v \in \cO' \text{ and } u \in \cO'\setminus C' $} \\
0 & \mbox{otherwise} \\
\end{array}
\right.
\]

Note that this means we are assigning zero benefit to the edges whose both vertices are in $C'$ to avoid duplication of benefits assigned to edges within $C(I)$. The total benefit of $E(G')$ is the same as the total benefit of $E(G)$.

\item
$w'(v') = w(v) , \ \forall v' \in C'$, i.e.
the copies of the vertices in $C'$ have the
same weight on each vertex as in $C(I)$.
All other vertices $v$ have the same weight in $G'$ as in $G$.

\end{enumerate}

We consider two cases:
\begin{enumerate}

\item[Case 1]
$ b(C(I)) \geq  b(\cO(I))/2$.
In this case $b(A(I)) \geq b(C(I)) \geq b(\cO(I))/2$. Thus $\frac{b(\cO(I))}{b(A(I))} \leq 2$.

\item[Case 2]
$ b(C(I)) <  b(\cO(I))/2$.
Here, $ b'(\cO') = b(\cO (I)) - b(C(I)) > b(\cO(I))/2$. We consider this case in more detail below.

\end{enumerate}

Let $A(I')$ be an approximate solution provided by the greedy algorithm in the instance
$I'$. We show that there exists a sequence of choices made by the greedy such
that $A(I')=A(I)$.

\begin{lemma}
For solution $A(I)$ provided by the greedy algorithm  in instance $I$, the greedy
algorithm can also produce the  solution $A(I)$ on instance $I'$.
\end{lemma}
\begin{proof}

We show that at every step of the greedy algorithm on $I'$, the partial solution produced is
the same as that in $I$. Consider the $i$th step of the greedy algorithm, assuming that the claim is true for prior steps. Prior to that step, let $A_i(I')$ be the set of vertices chosen by the greedy in $I'$. By assumption, prior to the $i$th step, the same set of vertices in $A(I)$ will be chosen by the greedy algorithm on $I$ also.

Let $V_i$ be the vertices chosen at step $i$ in $I'$. Note that the Greedy Algorithm does not need to pick anything from $\cO \setminus C'$ since there is as good or better choice in $A(I')$ as per the greedy choice in the instance $I$. Thus, $V_i \subset A(I') \cup C'$.

Let $C_i(I')$ be the set of vertices in $C(I)$ corresponding to $V_i\cap C'$,
and let $V_i'=(V_i-C')\cup C_i(I')$.

Then $A_i(I')\cup V_i'$ has benefit greater or equal to the benefit of
$A_i(I')\cup V_i$ (since benefits are zero within $C'$), and the same or less weight as $A_i(I')\cup V_i$
(less if $C_i(I')$ intersects $V_i\cap C(I)$).  Therefore $V_i'$ is another valid
choice for the greedy algorithm on $I'$ at the $i$th step.  Since $V_i'  \subseteq A(I)$,
the greedy algorithm on $I'$ always has a choice within $A(I)$.

\end{proof}

Given that $b'(\cO') > \frac{b(\cO(I))}{2}$, by Lemma 3.1 and consequently $b(A(I) = b'(A(I'))$, it follows that
\[ \frac{b(A(I))}{b(\cO(I))} \geq \frac{b'(A(I'))}{2b'(\cO')} \]

We now bound the ratio $\frac{b'(A(I'))}{b'(\cO')}$. Let $V_i(I')$ be the set
of vertices chosen at the $i$th step of the greedy algorithm. And let
$K_i = \frac{b'(V_i(I'))}{w'(V_i(I'))} $. 
Note that $b'(A(I')) = {\sum^t}_{i=1} b'(V_i(I'))$ and 
$w'(A(I'))= {\sum^t}_{i=1} w'(V_i(I'))$, i.e. $b'(V_i(I'))$ is the sum
of the enefits of edges incident to vertices in $V_i(I')$
within ${\cup^i}_{j=1} V_j(I')$.
Further, let $K_{min} = \min_i \{ K_i , i = 1 \ldots r \}$
when the greedy executes for $r$ iterations. 

Let $K_{min} = \frac{b'(V_t(I'))}{w'(V_t(I'))}$. 
Then $\frac{b'(V_t(I'))}{w'(V_t(I'))} \leq  \frac{b'(V_i(I'))}{w'(V_i(I'))}$ $\forall i \ne t$. 

This means
$\sum_i w'(V_i(I')) b'(V_t(I')) \leq \sum_i w'(V_t(I')) b'(V_i(I'))$, that is $\frac{b'(V_t(I'))}{w'(V_t(I'))} \leq  \frac{\sum_i b'(V_i(I'))}{\sum_i {w'(V_i(I'))}}$. 
Thus, we have
  
\[ K_{min} \leq \frac{b'(A(I'))}{w'(A(I'))} \]

By the choice of the greedy, and the fact that $A(I') \cap \cO' = \emptyset$,
\[ b'(T) \leq K_{min} w'(T) , \ \ \forall T \subseteq \cO' , |T| = t \]

Note that $b'(T) = \sum_{e \in T} b'(e) $.
Adding up the inequalities for all such $T \subseteq \cO'$(Note that all such
$T$ are feasible) gives
\begin{equation}
{{|\cO'|-2} \choose {t-2}} \sum_{e \in \cO'} b'(e)  \leq
K_{min} {{|\cO'|-1} \choose {t-1}} \sum_{v \in \cO'} w'(v)
\end{equation}

Since $b'(\cO') = \sum_{e \in \cO'} b'(e)$

\[ {{|\cO'|-2} \choose {t-2}} b'(\cO') \leq K_{min} {{|\cO'|-1} \choose {t-1}} w'(\cO'), \]
which yields
\[ b'(\cO') < \frac{|\cO'|}{t-1} K_{min} w'(\cO') \]
and
\[ \frac{b(\cO(I))}{2} < b'(\cO')  < \frac{|\cO'|b(A(I'))}{(t-1)\cdot w(A(I'))} w(\cO') \]

Further, since
$w(\cO') \leq pW, b(A(I'))=b(A(I))$ and $w(A(I')) = w(A(I)) \geq W/2$ we get
\[ \frac{b(\cO(I))}{b(A(I))} \leq  \frac{2}{t-1} |\cO(I)| \frac{pW}{w(A(I))} \leq  \frac{4p}{t-1} |\cO(I)| \]

Since the weight $w(v)$ of each vertex is assumed to be positive (and hence at least 1 under our integrality assumptions), $|\cO(I)| \leq W$ and $ n$, we get the two bounds $\frac{4pn}{t-1}$ and $\frac{4pW}{t-1}$.\\

In the end, we obtain (by change of parameter from $t-1$ 
with $t \ge 2$ to $t'$ with $t' \ge 1$) and using lemma~{}:
\begin{theorem}
\label{greedythm}
For a fixed $t\geq 1$, the greedy algorithm gives a $(8p \min(n,W)/t)$-factor
polynomial time ($O(2^{t+1}{{n+1} \choose {t+1}})$-running time) approximation algorithm
for {\bf PIQPR, PIQPS} and hence for {\bf PIQP}, where $W = \max_i W_i$.

\end{theorem}

\section{Randomized  Algorithm }

In this section we describe a randomized approximation method for the
PIQP problem. We first describe a randomized algorithm to provide an
approximation bound and then  improve that bound
by combining with the bound obtained via the
greedy algorithm.
\subsection{A Randomized Method}

Our main result in this section is:

\begin{theorem}
\label{thmmain} 
PIQP with at most logarithmic number of linear constraints ($p \le \log n$) and $\min_i{W_i} \geq \beta(n)$ can be approximated to within a factor
of $O(a_{max} \frac{n}{\beta(n)}\log^{2 + \gamma} n)$ w.h.p. , where $a_{max}$ is the maximum size of an entry in the constraint matrix $A$, and $\gamma >0$ is an arbitrary small constant.
\end{theorem}

The randomized solution that we will use for $PIQP$ 
will be the solution
with maximum benefit among three different solutions, two deterministic
and one randomized.
We will use the Greedy Algorithm from Section~3, to get a solution within an approximation factor of $O(\min\{n,W\})$,
as indicated by Theorem \ref{greedythm}, to $PIQP$. Note that the approximation provided by the greedy
implies that we can assume $a_{max} \leq \beta(n)$. If $a_{max} > \beta(n)$ then $n < a_{max} n/ \beta(n)$ and hence
the Greedy Algorithm gives a better approximation factor than claimed by Theorem~4.1.

Now, let us describe the three potential solutions.
For any solution $x \in \{0,1\}^n$, we define the functions values $F(x)$ and $G_i(x)$ as the
objective function value and the constraint usage values, respectively, of $x$. That is, $F(x) = \sum\limits_{i,j} b_{ij}x_i x_j$
and $G_i(x) = \sum\limits_{j} a_{ij} x_j$. As in Section~3, we will interpret 
as an instance of GKPM, the multiconstraint version of GKP, with each variable $x_i$ corresponding to a vertex $v_i$ and $b_{ij}$ corresponding to the benefit on the edge $v_iv_j$. The coefficient $a_{ij}$ is the weight (cost) associated with $v_j$ in the $i$th constraint. To simplify the notation, sometimes we will use vertex notation $u$ in the subscripts, as in $b_{uv}$ (the benefit on the edge $uv$) and $a_{i,u}$ (weight in the constraint $i$ associated with vertex $u$), since each such subscript $i$ corresponds to both a variable and a vertex.

\paragraph{Solution 1:} 

Note that by the form of {\bf PIQP} 
$ \forall i$ we can assume that $a_{ij} + a_{ik} \leq W$ for every pair of indices $(j,k)$.
We can pick a solution $x_i =x_j =1$ corresponding to the edge $v_iv_j$ with 
maximum value of $b_{ij}$ over all possible edges $v_iv_j$. 
This will be out first solution.

\paragraph{Solution 2:}
For the remaining two solutions, we need to use a ``hyperbolic" program ($P^*$):

$$\begin{array}{rlll}
&maximize & \sum\limits_{uv \in E(G)} b_{uv}x_{uv}\\
&\text{such that}&
AX \leq We&\\
&&x_ux_v\geq x_{uv}^2\\
&&x_u, x_{uv} \in [0,1]\\
\end{array}$$

\paragraph{Solution 2(a)}
This optimization problem can be solved in polynomial time \cite{LVBL98} to get an optimal solution $x_u^* \in [0,1]$.
Let $Y$ be a random 0-1 solution generated by $Y_u = 1$ with probability $\sqrt{x_u^*}/\lambda$, where $\lambda = 2 \sqrt{a_{max} n / \beta(n)}$ is a scaling factor.

\begin{lemma}
\label{mainlemma1}
Let $W \geq \beta(n)$.
Then for each $i$,
$\sum_u a_{i,u} \sqrt{x_u^*} \leq 2 W\sqrt{a_{max} n/\beta(n)}$, where $a_{max} = \max\{a_{ij}\}$.
\end{lemma}

\begin{proof}
Fix an $i$. Let $I = \{ u | x_u > \beta(n)/ (a_{iu} n)\}$ and $J = \{ u | x_u \leq \beta(n)/ (a_{iu} n)\}$.
Then, using the fact that $\sqrt{x_u} < x_u \sqrt{a_{iu} n/ \beta(n)}$ when $u \in I$, we get

$$\begin{array}{rcl}
\sum_u a_{iu} \sqrt{x_u}& = &\sum_{u \in I} a_{iu} \sqrt{x_u} + \sum_{u \in J} a_{iu} \sqrt{x_u} \\
& \leq & \sum_{u \in I} a_{iu} x_u \sqrt{a_{iu}  n / \beta(n)}
+ \sum_{u \in J} a_{iu} \frac{1}{\sqrt{a_{iu} n / \beta(n)}}\\
& \leq & \sqrt{a_{max} n / \beta(n)} \sum_{u \in I} a_{iu} x_u
+ \sqrt{a_{max}} \sum_{u \in J}\frac{1}{\sqrt{n / \beta(n)}}\\
& \leq & \sqrt{a_{max}  n / \beta(n)} W +  \sqrt{a_{max}} \frac{n}{\sqrt{n / \beta(n)}}\\
& \leq & \sqrt{a_{max}} W( \sqrt{n / \beta(n)} + \frac{n}{\beta(n) \sqrt{n / \beta(n)}})\\
& \leq & 2 \sqrt{a_{max}} \sqrt{n / \beta(n)} W
\end{array}$$

\end{proof}

Taking $\lambda = 2 \sqrt{a_{max} n / \beta(n)}$, gives us that
${\bf E}[G_i(Y)] = \sum a_{iu} \sqrt{x_u^*}/\lambda \leq 2 \sqrt{a_{max} (n/\beta(n))} W/ \lambda \leq W$, using the lemma.
Note that ${\bf E}[F(Y)] = \sum b_{uv} \sqrt{x_u^*} \sqrt{x_v^*}/{\lambda}^2$ equals the optimal value of the hyperbolic program ($P^*$) divided by $\lambda^2$, since $x_u^*x_v^* = x_{uv}^*$ at optimality.

\paragraph{Solution 2(b):}
Let $Z$ be a 0-1 solution generated by the following procedure.
First find $v \in V(G)$ such that $v =\arg\!\max_w \sum_{u \in N_G(w)} b_{uw} \sqrt{x_u^*}/\lambda$. Then, for this fixed $v$, define $Z_v = 1$, $Z_w = 0$ at every vertex, $w \not\in N_G(v) \cup \{v\}$ ($v$ and its neighbors in $G$), and $Z_u$ for $u \in N_G(v)$  is determined by solving a ``local" 0-1 Knapsack problem with multiple constraints, whose items are the neighbors of $v$ and benefit of each such item equals the benefit of the edge incident to it and $v$,
and is solvable by a PTAS \cite{books/knapsackProblems}:

$$\begin{array}{rlll}
& \max & \sum\limits_{e: e= vu} b_{e}z_{u}\\
&\text{such that}&
\sum\limits_{e: e= vu} a_{iu}z_u \leq W - a_{iv},\;\; i=1, \ldots, p&\\
&& z_u \in \{0,1\} \\
\end{array}$$


Thus, $G_i(Z) = \sum\limits_{e: e= vu} a_{iu}z_u +  w(v) \leq W$, that is $Z$ is a feasible solution (recall al other variables are 0). Also note that $F(Z) =  \sum\limits_{e: e= vu} b_{e}z_{u}$.

We claim that picking the best of the solutions obtained
will give a solution with the stated  approximation factor.

We have to show that the function $F$ evaluated at this 0-1 solution is not too far from the optimal solution for {\bf PIQP scaled}  with high probability, and similarly each of the functions $G_i$ evaluated at this 0-1 solution satisfy the budget bound with high probability.

\subsection{Analyzing the Objective Function}
\label{objanalysis}

Let us first study the function $F(x) = \sum\limits_{uv \in E(G)} b_{uv}x_u x_v$.

For $0<\alpha < 1$, with $P$ denoting the 0-1 hyperbolic program for solving QKP, we have that
$$\begin{array}{rl}
&{\bf P}[F(Y) < (1-\alpha) OPT(P)/\lambda^2]\\
\leq& {\bf P}[F(Y) < (1-\alpha) OPT(P^*)/\lambda^2]\\
=& {\bf P}[F(Y) < (1-\alpha) {\bf E}[F(Y)] ]\\
=& {\bf P}[{\bf E}[F(Y)] - F(Y) > \alpha {\bf E}[F(Y)]  ]\\
\end{array}$$

If we can show that this probability is small, then that would prove that $Y$ is within a factor $\lambda^2/(1-\alpha)$ of the optimal (as long as the budget constraints are satisfied). \\

Recall that $\varepsilon= \max \{\varepsilon_0, \varepsilon_1, \varepsilon_2 \}$, and $\varepsilon'= \max \{\varepsilon_1, \varepsilon_2 \}$.
Here $\varepsilon_0 = {\bf E}[F(Y)]= OPT(P^*)/\lambda^2$, $\varepsilon_1 = \max_v (\sum_{u\in N_G(v)}b_{uv} {\bf P}[Y_u =1]) = \max_v (\sum_{u\in N_G(v)} b_{uv} \sqrt{x_u^*}/\lambda)$, $\varepsilon_2 = \max_{uv\in E(G)} b_{uv}$.\\

We split the analysis into four cases depending on the relative worth of the expected value of the solution $Y$ with regard to $Z$ and the edge with maximum benefit (Solution 1). Let $\gamma > 0$ be a fixed small constant.

\paragraph{Case (i)} When $\varepsilon_2 > \varepsilon_1$ and $\varepsilon_0 < \varepsilon_2 \log^{2+\gamma} n$, Solution 1, $\max_{uv} b_{uv}$ provides a good solution.\\

By assumption, $\max_{uv\in E(G)} b_{uv}$ is greater than $\max_v (\sum_{u\in N_G(v)} \sqrt{x_u^*}/ \lambda$  as well as\\ $\sum b_{uv} \sqrt{x_u^*} \sqrt{x_v^*}/ \lambda^2 \log^{2+\gamma} n)$, i.e., $\max_{uv} b_{uv} > OPT(P)/\lambda^2$. So  $\max_{uv} b_{uv}$ is within a factor $\lambda^2 \log^{2+\gamma} n$ of the optimum.\\

\paragraph{Case (ii)} When $\varepsilon_2 > \varepsilon_1$ and $\varepsilon_0 > \varepsilon_2 \log^{2+\gamma} n$, the concentration bounds provided by
the results in  Kim-Vu \cite{kimvu:concentration, vu:con2} bound the error in the randomized solution $Y$.\\

From the Kim-Vu bound,  ${\bf P}[{\bf E}[F(Y)] - F(Y) > t^2] < 2e^2 \;\; exp\left(-t/32(2\varepsilon\varepsilon')^{1/4} + \log n\right)$. Taking $t^2 = \alpha {\bf E}[F(Y)]$, we get

$$\begin{array}{rl}
&{\bf P}[{\bf E}[F(Y)] - F(Y) > \alpha {\bf E}[F(Y)]  ]\\
\leq & 2e^2 \;\; exp\left(\frac{-\sqrt{\alpha}}{32 2^{1/4}} \sqrt{\varepsilon_0}/(\varepsilon\varepsilon')^{1/4} + \log n\right)\\
\end{array}$$

Under the given assumptions, $\varepsilon = \varepsilon_0/ \log^{2+\gamma} n$ and $\varepsilon' = \varepsilon_2$, the above bound reduces to

$$\begin{array}{rl}
&2e^2 \;\; exp\left(\frac{-\sqrt{\alpha}}{32 \cdot 2^{1/4}} (\frac{\varepsilon_0}{\varepsilon_2} \log^{2+\gamma} n)^{1/4} + \log n\right)\\
\leq &2e^2 \;\; exp\left(\frac{-\sqrt{\alpha}}{32 \cdot 2^{1/4}}(\log^{2+\gamma} n \log^{2+\gamma} n)^{1/4} + \log n\right)\\
= & 2e^2 \;\; exp\left(\frac{-\sqrt{\alpha}}{32 \cdot 2^{1/4}} \; \log^{1 + \frac{\gamma}{2}} n + \log n\right)\\
\end{array}$$

which is o(1). \\

\paragraph{Case (iii)} When $\varepsilon_1 > \varepsilon_2$ and $\varepsilon_0 > \varepsilon_1 \log^{2+\gamma} n$,  the bounds provided by Kim-Vu \cite{kimvu:concentration, vu:con2}
can again be applied  to bound the error in  $Y$, the randomized solution.\\

The application of  the concentration bounds from Kim-Vu to $F(Y)$ is
identical to the previous case with the roles of $\varepsilon_1$ and $\varepsilon_2$ interchanged.

\paragraph{Case (iv)} When $\varepsilon_1 > \varepsilon_2$ and $\varepsilon_0 < \varepsilon_1 \log^{2+\gamma} n$,  we show that the solution $Z$, obtained from the
local multi-constrained knapsack problem works as a good solution.\\


By the definition of $Z$, $F(Z)$, the objective function value obtained from
the integer solution to the knapsack problem by rounding
the fractional solution is within a factor of $p+1$ of
the optimal fractional solution 
provided by $ \max \sum\limits_{e: e= vu} b_{e}z'_{u} $
subject to 
$\sum_{u: (v,u) \in E} a_{iu}z'_u \leq W-w(v), \  \forall i , z''_{u} \in [0,1],$.
This bound is proved in the appendix.

Since $w(v) \leq W/2$, this fractional solution is within a factor of $2$ of the solution to
the relaxation, $RKP$:
\[  \max \sum\limits_{e: e= vu} b_{e}z''_{u} \; \;\]
\[{\rm s. \ t.}
\sum_{u: (v,u) \in E} a_{iu}z''_u \leq W,  \ \forall i\]
\[ z''_{u} \in [0,1] \]

Since $\varepsilon_1$ is the objective function value evaluated at the 
neighborhood of $v$, $\sum_{u\in N_G(v)} b_{uv} \sqrt{x_u^*}/\lambda$, 
at a particular solution $x_u^*$ of $P^*$, and
by Lemma~3.1, such a solution also satisfies the 
constraint $\sum_u a_{iu} \sqrt{x_u^*}/\lambda \leq W$, 
we conclude that $ \sum\limits_{e: e= vu} b_{e}z^{opt}_{u}  \geq \varepsilon_1$, 
where $z^{opt}$ is the optimal solution to the relaxation $RKP$.

Since $\varepsilon_1$ is at least $OPT(P)/ \lambda^2 \log^{2+ \gamma} n$,  
$Z$ is within a factor $O(\lambda^2 \log^{2+\gamma} n)$ of the optimum. 
Note that $Z$ satisfies the budget constraint.\\

\subsection{Analyzing the Budget Function}
\label{budgetanalysis}

We consider the function $G_i(Y)$.
We know by the choice of $\lambda$ that \[E[G_i(Y)]  \leq W \]

\ignore{
Further note that since our program is convex, the  weight constraint is tight. Thus
\[\sum_u a_{iu}x^*(u) \geq W \]
and
\[ {\bf E}[G_i(Y)] = \sum_u a_{iu} \sqrt{x_u^*}/\lambda \geq \sum_u a_{iu}x^*(u)/\lambda \geq W/\lambda \]

By the Chernoff-Hoeffding bound we obtain:
$$\begin{array}{rl}
\mPr [G_i(Y) > (1 +\delta) W ]& \\
\leq& {\bf P} [ G(Y) > (1+\delta)E(G(Y)) ] \\
\leq&  e^{-({\bf E}[G(Y)] \delta^2)/(3 a_{max})}\\
\leq& e^{-(\beta(n) \delta^2)/ (\lambda 3a_{max}) }\\
\leq& e^{-((\delta^2/6 \sqrt{n}) (\beta(n)/ a_{max})^{3/2})} \\
\leq& e^{-((\delta^2/6 \sqrt{n}))}\\
\end{array}$$
using the assumption that $ \beta(n) \geq  a_{max}$.

}

\begin{claim} $\mPr [G_i(Y) > (1 +\delta) W ] = o(1) \forall i$  
\end{claim}

\paragraph{Case 1:} ${\bf E}[G_i(Y)] \geq W/\lambda$

By Chernoff-Hoeffding Bound-1 we obtain:
$$\begin{array}{rl}
\mPr [G_i(Y) > (1 +\delta) W ]& \\
\leq& \mPr [ G_i(Y) > (1+\delta)E(G_i(Y)) ] \;\;\; {\text since} \;\;{\bf E}(G_i(Y)) \leq W\\
\leq&  e^{-({\bf E}[G_i(Y)] \delta^2)/(3 a_{max})}\\
\leq& e^{-(\beta(n) \delta^2)/ (\lambda 3a_{max}) }\;\;\; \text{since} \;\;{\bf E}[G_i(Y)] \geq W/\lambda \geq \beta)n)/\lambda\\
\leq& e^{-((\delta^2/6 \sqrt{n}) (\beta(n)/ a_{max})^{3/2})} \;\;\; 
\text{using the value of } \lambda \\
\leq& e^{-((\delta^2/6 \sqrt{n}))}\\
\end{array}$$
using the assumption that $ \beta(n) \geq  a_{max}$.

\paragraph{Case 2:} ${\bf E}[G_i(Y)] \leq W/\lambda$

By Chernoff-Hoeffding Bound-2 
we obtain:
$$\begin{array}{rl}
\mPr [G_i(Y) > (1 +\delta) W ]& 
\leq \mPr [G_i(Y) > \frac{(1 +\delta)}{\lambda} W] \;\;\; \text{since} \;\; \lambda >0\\
\leq& \mPr [ G_i(Y) - \frac{W}{\lambda}> \delta \frac{W}{\lambda}] \\
\leq& \mPr [ G_i(Y) - E(G_i(Y)) > \delta \frac{W}{\lambda}] \;\;\; \text{since} \;\;{\bf E}(G_i(Y)) \leq W/\lambda\\
\leq&  e^{-(2 \delta^2 W^2/\lambda^2)/(n a_{max}^2)}\\
=& e^{-\frac{2\delta^2}{n}(\frac{W^2}{\lambda^2 a_{max}^2}) }\\
\leq& e^{-\frac{2\delta^2}{n}(\frac{W^2 \beta(n)}{4 a_{max}^3}) }\;\;\; 
\text{using the value of } \lambda \\
\leq& e^{-\frac{\delta^2}{2n}(\frac{\beta^2(n) \beta(n)}{a_{max}^3}) }\;\;\; 
\text{since } W \geq \beta(n) \\
\leq& e^{-\frac{\delta^2}{2n}} \;\;\; \text{since } \beta(n) \geq a_{max}\\
\end{array}$$

We note that to obtain a solution that satisfies the p+1 conditions:
\begin{enumerate}

\item
$ {\bf P}[{\bf E}[F(Y)] - F(Y) > \alpha {\bf E}[F(Y)]=o(1)$

and

\item
$\mPr [G_i(Y) > (1 +\delta) W ] = o(1) \forall i$
\end{enumerate}

we compute the probability that any one of these conditions is not satisfied.
We will simply concentrate on the second set of probabilities, since the first inequality
has already been shown to be $o(1)$.


If we repeat the randomized algorithm $k$ times this probability
becomes $e^{-(\delta^2 k/6 \sqrt{n})}$ and adding up all these
probabilities for the $p$ constraints simply
makes it $pe^{-(\delta^2 k/6 \sqrt{n})}$
which is $o(1)$ when $k = n^{\frac{1}{2} + \epsilon}$, when $p$
is a fixed number.
In fact, when $p= O(\log n)$,  $k=O(n^c)$, iterations,
for a constant $c$, suffices.

This gives us the approximation bound w.h.p.

\subsection{Approximating PIQP}


We use the result of the previous section  to provide an
approximation for the PIQP problem.

The greedy algorithm (Theorem~\ref{greedythm}
) provides an $O(\min (n,W))$  approximation factor method.
Note that if $W < n^{1/2}$, the greedy algorithm provides a
$O(\sqrt{n})$ factor algorithm.

Theorem~\ref{thmmain} proved in the previous section gives a
$O(a_{max} \frac{n}{\beta(n)})\log^{2 + \gamma} n)$ factor algorithm. Choosing $\beta(n)=W \geq n^{1/2}$, Theorem~\ref{thmmain} and Theorem~\ref{greedythm} give us the following result by picking the best of the two algorithms:

\begin{theorem}
PIQP can be approximated to within a factor of
$O(a_{max}\sqrt{n}log^{2+\gamma}n)$ w.h.p. in polynomial time, where $\gamma$
is an arbitrary small constant.
\end{theorem}

\bibliographystyle{plain}
\bibliography{proposal}

\begin{thebibliography}{10}

\bibitem{BCCFV2010}
Aditya Bhaskara, Moses Charikar, Eden Chlamtac, Uriel Feige, and Aravindan
  Vijayaraghavan.
\newblock Detecting high log-densities: an $o(n^{1/4})$ approximation for
  densest k-subgraph.
\newblock In {\em Proceedings of the 42nd ACM symposium on Theory of
  computing}, STOC '10, pages 201--210, New York, NY, USA, 2010. ACM.

\bibitem{BHW2009}
Glencora Borradaile, Brent Heeringa, and Gordon Wilfong.
\newblock The 1-neighbour knapsack problem.
\newblock In {\em Proceedings of the 22nd international conference on
  Combinatorial Algorithms}, IWOCA'11, pages 71--84, Berlin, Heidelberg, 2011.
  Springer-Verlag.

\bibitem{C52}
H.~Chernoff.
\newblock A measure of asymptotic efficiency for tests of a hypothesis based on
  the sum of observations.
\newblock {\em Annals of Mathematical Statistics}, 23(4), 1952.

\bibitem{FKP2001}
U.~Feige, G.~Kortsarz, and D.~Peleg.
\newblock The dense {$k$}-subgraph problem.
\newblock {\em Algorithmica}, 29(3):410--421, 2001.

\bibitem{GHS1980}
G.~Gallo, P.~L. Hammer, and B.~Simeone.
\newblock Quadratic knapsack problems.
\newblock In {\em Combinatorial Optimization}, volume~12 of {\em Mathematical
  Programming Studies}, pages 132--149. 1980.

\bibitem{GNY1994}
O.~Goldschmidt, D.~Nehme, and G.~Yu.
\newblock On the set-union knapsack problem.
\newblock {\em Naval Research Logistics}, 41(6):833--842, 1994.

\bibitem{GS2011}
Venkatesan Guruswami and Ali~Kemal Sinop.
\newblock Lasserre hierarchy, higher eigenvalues, and approximation schemes for
  graph partitioning and quadratic integer programming with psd objectives.
\newblock In {\em Proceedings of the 2011 IEEE 52nd Annual Symposium on
  Foundations of Computer Science}, FOCS '11, pages 482--491, Washington, DC,
  USA, 2011. IEEE Computer Society.

\bibitem{Hastad99}
Johan Hastad.
\newblock Clique is hard to approximate within $n^{1-\epsilon}$.
\newblock {\em Acta Mathematica}, 182:105--142, 1999.

\bibitem{H63}
W.H. Hoeffding.
\newblock Probability inequalities for sums of bounded random variables.
\newblock {\em Journal of the American Statistical Association}, 58(301), 1963.

\bibitem{RW2002}
David J.~Rader Jr. and Gerhard~J. Woeginger.
\newblock The quadratic 0-1 knapsack problem with series-parallel support.
\newblock {\em Oper. Res. Lett.}, 30(3):159--166, 2002.

\bibitem{KPP2004}
Hans Kellerer, Ulrich Pferschy, and David Pisinger.
\newblock {\em Knapsack problems}.
\newblock Springer-Verlag, Berlin, 2004.

\bibitem{books/knapsackProblems}
Hans Kellerer, Ulrich Pferschy, and David Pisinger.
\newblock {\em Knapsack problems.}
\newblock Springer, 2004.

\bibitem{Khot}
Subhash Khot.
\newblock Ruling out {PTAS} for graph min-bisection, dense k-subgraph, and
  bipartite clique.
\newblock {\em SIAM J. Comput}, 36:1025--1071, 2006.

\bibitem{kimvu:concentration}
Jeong-Han Kim and Van~H. Vu.
\newblock Concentration of multivariate polynomials and its applications.
\newblock {\em Combinatorica}, pages 417--434, 2000.

\bibitem{KimKojima2001}
S~Kim and M~Kojima.
\newblock {Second order cone programming relaxation of nonconvex quadratic
  optimization problems}.
\newblock {\em {OPTIMIZATION METHODS \& SOFTWARE}}, {15}({3-4}):{201--224},
  {2001}.

\bibitem{KS2007}
Stavros~G. Kolliopoulos and George Steiner.
\newblock Partially ordered knapsack and applications to scheduling.
\newblock {\em Discrete Appl. Math.}, 155(8):889--897, 2007.

\bibitem{LiKK11}
Z.~Li, S.~Kapoor, H.~Kaul, E.~Veliou, B.~Zhou, and C.~Lee.
\newblock A new methodology for transportation investment decisions considering
  project interdependencies.
\newblock {\em Journal of the Transportation Research Board}, pages 36--46,
  2012.

\bibitem{lbvl98}
M.~Lobo, L.~Vandenberghe, S.~Boyd, and H.~Lebret.
\newblock Applications of second-order cone programming.
\newblock {\em Linear Algebra and its Applications}, 284:193--228, 1998.

\bibitem{LVBL98}
Miguel~S. Lobo, Lieven Vandenberghe, Stephen Boyd, and Herv\'{e} Lebret.
\newblock {Applications of second-order cone programming}.
\newblock {\em Linear Algebra and its Applications}, 284(1-3):193--228,
  November 1998.

\bibitem{PS2009}
Ulrich Pferschy and Joachim Schauer.
\newblock The knapsack problem with conflict graphs.
\newblock {\em J. Graph Algorithms Appl.}, 13(2):233--249, 2009.

\bibitem{P2007}
D.~Pisinger.
\newblock The quadratic knapsack problem - a survey.
\newblock {\em Discrete Appl. Math.}, 155:623--648, 2007.

\bibitem{SY2000}
N.~Samphaiboon and T.~Yamada.
\newblock Heuristic and exact algorithms for the precedence-constrained
  knapsack problem.
\newblock {\em J. Optim. Theory Appl.}, 105(3):659--676, 2000.
\newblock Special Issue in honor of Professor David G. Luenberger.

\bibitem{SW1998}
Anand Srivastav and Katja Wolf.
\newblock Finding dense subgraphs with semidefinite programming.
\newblock In {\em Approximation algorithms for combinatorial optimization
  ({A}alborg, 1998)}, volume 1444 of {\em Lecture Notes in Comput. Sci.}, pages
  181--191. Springer, Berlin, 1998.

\bibitem{S99findingdense}
Anand Srivastav and Katja Wolf.
\newblock Finding dense subgraphs with mathematical programming, 1999.

\bibitem{vu:con2}
V.~H. Vu.
\newblock Concentration of non-lipschitz functions and applications.
\newblock {\em Random Structures and Algorithms}, 20(3):262--316, 2002.

\bibitem{YKW2002}
Takeo Yamada, Seiji Kataoka, and Kohtaro Watanabe.
\newblock Heuristic and exact algorithms for the disjunctively constrained
  knapsack problem.
\newblock {\em IPSJ J.}, 43(9):2864--2870, 2002.

\bibitem{Zuckerman05}
David Zuckerman.
\newblock Linear degree extractors and the inapproximability of max clique and
  chromatic number.
\newblock In {\em Proceedings of the thirty-eighth annual ACM symposium on
  Theory of computing}, STOC '06, pages 681--690, New York, NY, USA, 2006. ACM.

\end{thebibliography}

\section{Appendix}

Consider the  multi-constrained knapsack problem, $MKP$:
\[  \max B(X)= \sum\limits_{j} b_j x_j \]
\[{\rm s. t.} 
\sum_{j} a_{ij}x_j \leq W,  \  i= 1 , \ldots p \]
\[ x_j \in [0,1], j = 1, \ldots n \]
We will also refer to the budget constraints by the inequality
system \[ AX \leq W \]
where $A \in {\cal R}^{p \times n}$ and $X \in {\cal R}^n$.
W.l.o.g.  we will assume that $a_{ij} \leq  W$.
An approximate solution to this optimization problem can
be obtained by solving the fractional linear program obtained
by relaxing the integrality constraints on $x_j$, to be simply
$ 0 \leq x_j \leq 1$.  Let the optimum benefit
obtained be $OPT_R(MKP)$. And let the optimum of $MKP$ be
denoted by $OPT(MKP)$.  Clearly, $OPT(MKP) \leq OPT_R(MKP)$

Note that in the relaxed LP, the polytope defining the feasible region
has vertices defined by $n$ inequalities, that must be satisfied at 
the vertex. At most $p$ of these are from the budget constraints, $AX \leq W$. 
The other inequalities that are satisfied specifiy integral values
for the variables $x_j$. Thus at most $p$ components in  the vector
$X$ ae non-integral. Let us denote the corresponding varibles by $X_p$.
Let $x_m$ be the variable such that $x_m= \arg \max_{x_j \in X_p} b_jx_j$. 
Construct an integral solution, $X_I$, by setting the variables in $X-X_p$ according 
to the integral solution
in the relaxed LP. Let $B(X_I)$ be the benefit obtained from the integral solution.
Furthermore, consider the solution, $X_M$, where $x_m=1$ and
all other variables in $X_p$ are assigned the value $0$.
And let $B(X_M)$ be its benefit.
The benefit obtained 
by $\max \{ B(X_M), B(X_I) \} ) \geq OPT_R(MKP)/(p+1) \geq OPT(MKP)/(p+1)$.

This gives us the following result:
\begin{theorem}

\end{theorem}
There exists an approximate solution, $A$, to the multi-constrained knapsack 
problem such that the benefit of the optimal solution is within a factor $p+1$
\end{document}